\documentclass[conference]{IEEEtran}
\IEEEoverridecommandlockouts
\usepackage{tikz}
\usetikzlibrary{arrows.meta,positioning}
\usepackage{balance}
\usepackage{cite}
\usepackage{amsmath,amssymb,amsfonts}
\usepackage{algorithmic}
\usepackage{graphicx}
\usepackage{textcomp}
\usepackage{xcolor}
\def\BibTeX{{\rm B\kern-.05em{\sc i\kern-.025em b}\kern-.08em
    T\kern-.1667em\lower.7ex\hbox{E}\kern-.125emX}}
\usepackage{amsmath,amssymb,amsthm,mathtools}
\usepackage{bm}
\usepackage[draft]{hyperref}
\usepackage[nameinlink]{cleveref}
\usepackage{array}
\newtheorem*{example*}{Example}
\usepackage{amsthm,thmtools}
\usepackage{tabularx,makecell}
\usepackage{amsmath, amsthm, amscd, amsfonts, amssymb, graphicx, color, booktabs, multirow, rotating, tikz}
\usepackage{booktabs}   
\usepackage{lscape} 
\usepackage{multirow}   
\usepackage{tabularx}
\usepackage{float}
\usepackage{array} 
\newcolumntype{C}[1]{>{\centering\arraybackslash}p{#1}}
\usepackage{stfloats}

\newtheorem{theorem}{Theorem}
\newtheorem{lemma}{Lemma}
\newtheorem{definition}{Definition}
\newtheorem{remark}{Remark}
\newtheorem{example}{Example}
\newtheorem{proposition}{Proposition}

\usepackage{array,tabularx} 
\newcolumntype{Y}{>{\raggedright\arraybackslash}X}
\newcolumntype{L}{>{\raggedright\arraybackslash}X}

\renewcommand{\baselinestretch}{0.97}
\usepackage{geometry}
\begin{document}

\title{Lifted-Product QLDPC Codes\\ in the Polynomial Domain

}\vspace{-1mm}

\author{\IEEEauthorblockN{Vahid Nourozi, Mohsen Moradi, Roxana Smarandache and David G. M. Mitchell}
}

\author{
Vahid~Nourozi$^\dag$,
Mohsen~Moradi$^*$, Roxana~Smarandache$^\ddag$,  
and David~G.~M.~Mitchell$^\dag$,
\\[2ex]
$^\dag$Klipsch School of Electrical and Computer Engineering,
New Mexico State University\\
Las Cruces, NM 88003, USA,
Email: nourozi@nmsu.edu; dgmm@nmsu.edu.\\
$^*$School of Electrical, Computer and Energy Engineering,
Arizona State University\\
Tempe, AZ 85287, USA,
Email: mmorad11@asu.edu. \\
$^\ddag$Departments of Mathematics and Electrical Engineering, University of Notre Dame\\
Notre Dame, IN 46556, USA,
Email: rsmarand@nd.edu.\vspace{-1mm}
}

\maketitle

\begin{abstract}
This paper presents a finite-length polynomial-domain formulation of lifted-product quantum low-density parity-check (QLDPC) codes. We formulate the code construction over the quotient ring $\mathbb{F}_2[D]/(D^L+1)$, where polynomial base matrices are lifted entrywise to binary circulant blocks. This representation gives a compact algebraic description of the lifted-product parity-check matrices and allows CSS orthogonality to be analyzed before binary expansion. We show that the standard circulant lifting map is compatible with polynomial conjugation, which implies that the resulting binary matrices satisfy the CSS commutation constraint. 
The construction is illustrated with a constraint length $7$, rate $1/2$ NASA convolutional code example, and numerical examples are provided from a $3\times 4$ polynomial parity-check matrix. The finite-length performance of selected constructed codes is then evaluated over the depolarizing channel using various benchmark decoders. The resulting framework gives a structured method to construct finite-length lifted-product QLDPC codes from small polynomial base matrices.
\end{abstract}



\section{Introduction}

\IEEEPARstart{Q}{uantum} error correction is essential for reliable quantum information processing in the presence of noise and decoherence. Within this area, Calderbank--Shor--Steane (CSS) codes provide one of the main algebraic frameworks to construct stabilizer codes from classical linear codes \cite{calderbank1996good,steane1996error, gottesman1997stabilizer}. Among modern CSS-based families, quantum low-density parity-check (QLDPC) codes have attracted particular attention because they offer sparse parity-check operators together with the possibility of nonzero asymptotic rate and growing minimum distance \cite{mackay2004sparse, breuckmann2021balanced}. A foundational step in this direction was the hypergraph-product (HGP) construction of Tillich and Z\'emor, which showed how tensor-product structure can be used to obtain QLDPC codes with positive rate and minimum distance proportional to the square root of the blocklength \cite{tillich2013quantum}.

More recently, Panteleev and Kalachev introduced the lifted-product (LP) framework, which generalizes product constructions for quantum codes  and led to quantum LDPC codes with almost linear growth of minimum distance \cite{panteleev2022quantum}. In subsequent work, they showed that lifted products over non-abelian groups yield asymptotically good quantum LDPC codes, thereby establishing a major breakthrough in the theory of QLDPC constructions \cite{panteleev2022asymptotically}. These results make the lifted product one of the central modern tools for building structured quantum codes. At the same time, practical implementations often require finite-length codes with explicit algebraic descriptions, controlled block size, and constructions that remain compatible with structured hardware-oriented representations.

One natural way to realize such structure is through polynomial and quasi-cyclic descriptions. Classical LDPC block and convolutional codes based on circulant matrices provide a convenient algebraic language for lifting small base matrices into large sparse binary parity-check matrices \cite{tanner2004ldpc, smarandache2012quasi, kovalev2013quantum}. In the LP setting, this point of view is especially useful because it allows the base matrices to be written over the quotient ring $\mathbb{F}_2[D]/(D^L+1)$ and then lifted entrywise to binary circulant blocks. Recent work on finite-length LP QLDPC codes has shown that, in the quasi-cyclic setting,  combinatorial constraints on the base matrices can strongly affect the resulting quantum minimum distance and degeneracy properties \cite{raveendran2025minimum, panteleev2021degenerate}. 
This highlights the need for an  algebraic formulation of the LP construction that makes the orthogonality mechanism transparent in the polynomial domain.

In this paper, we explore the finite-length polynomial-domain formulation of LP-QLDPC codes. In particular, we formulate the LP parity-check matrices directly over the quotient ring $R_L=\mathbb{F}_2[D]/(D^L+1)$, so that small polynomial base matrices provide compact descriptions of the corresponding circulant lifted binary parity-check matrices.  We then prove that standard circulant lifting is conjugate-compatible and use this property to give a compact polynomial-domain proof of the CSS commutation condition $H_XH_Z^T=0.$ This framework also connects convolutional code polynomial parity-check matrices with finite-length LP-QLDPC construction through tail-biting circulant lifting. This provides a polynomial domain construction and verification framework that can be used together with existing finite-length search and distance-optimization methods. 
The approach is illustrated with a constraint length $7$, rate $R=1/2$ NASA convolutional code example, and additional numerical examples are provided from a $3\times 4$ polynomial parity-check matrix. The finite-length performance of selected constructed codes is then evaluated over the depolarizing channel using various benchmark decoders. 


\section{Background}\label{sec:preliminaries}

\subsection{Polynomial Rings and Quotient Structures}
\begin{definition}[Polynomial Ring]
    The polynomial ring $\mathbb{F}_2[D]$ consists of all polynomials in the indeterminate $D$ with coefficients in the field $\mathbb{F}_2$, with addition and multiplication defined by standard polynomial arithmetic modulo 2.
\end{definition}

\begin{definition}[Quotient Ring]
    For a positive integer $L$, the quotient ring $R_L \triangleq  \mathbb{F}_2[D]/(D^L+1)$ consists of equivalence classes of polynomials in $\mathbb{F}_2[D]$, where $p(D) \sim q(D)$ if and only if $p(D) - q(D)$ is divisible by $D^L+1$.

By the polynomial division algorithm, every equivalence class has a
unique remainder of degree less than $L$. We use this remainder as
the canonical representative of the class. Hence, every element of
$R_L$ can be written uniquely as
$$p(D)=\sum_{i=0}^{L-1} a_iD^i\in \mathbb{F}_2[D]/(D^L+1),
\qquad a_i\in\mathbb{F}_2.$$
Since the field has characteristic two, the relation
$D^L+1=0$ in $R_L$ implies $D^L=1$. Hence, exponents are reduced
modulo $L$, and $D^{-1}=D^{L-1}$ in $R_L$.

\end{definition}





\begin{definition}[Polynomial conjugation]
\label{def3}
For a polynomial  $p(D)=\sum_{i=0}^{L-1} a_i D^i\in R_L$, define its conjugate (or involution) polynomial $p(D)^{*}\in R_L$ as
\[
p(D)^{*}\triangleq p(D^{-1}) =\sum_{i=0}^{L-1} a_i D^{-i}
= \sum_{i=0}^{L-1} a_i D^{(L-i)\bmod L}.
\]
\end{definition}

\begin{definition}[CSS Code]
An [[N,K,d]] CSS quantum code is specified by
two binary parity-check matrices \(H_X,H_Z\in \mathbb{F}_2^{r_X\times N}\)
and \(\mathbb{F}_2^{r_Z\times N}\) satisfying the CSS commutation condition
\[
H_XH_Z^T=0 .
\]
Equivalently, $\operatorname{Im}(H_X)\subseteq \ker(H_Z) \text{ and } 
\!\operatorname{Im}(H_Z)\subseteq \ker(H_X).$
The number of logical qubits is $K=N-\operatorname{rank}(H_X)-\operatorname{rank}(H_Z),$
and the quantum minimum distance is $d=\min\{d_X,d_Z\},$
where
\[
d_X=\min\{\operatorname{wt}(x)|x\in \ker(H_Z)\setminus
\operatorname{Im}(H_X)\},
\]
and
\[
d_Z=\min\{\operatorname{wt}(z)|z\in \ker(H_X)\setminus
\operatorname{Im}(H_Z)\}.
\]

\end{definition}

\subsection{Circulant Matrices and Lifting}

In this section, we formalize the basic objects that we will use repeatedly: circulant permutation matrices, the lifting operation from polynomials to circulant matrices, and the resulting polynomial–circulant correspondence \cite{tanner2004ldpc}.

\begin{definition}[Circulant permutation matrix]\label{def:circulant}
Let $L\in\mathbb{N}$ and $0\le s < L$.  
The $L\times L$ \emph{circulant permutation matrix} $\mathrm{Circ}_L(s)$ is defined as the cyclic shift of the identity matrix row-by-row by
\(s\) positions to the right, i.e., for all $i,j\in\{0,1,\dots,L-1\}$,
\[
  \bigl(\mathrm{Circ}_L(s)\bigr)_{i,j}
  = \begin{cases}
      1, & \text{if } j \equiv i + s \pmod L,\\
      0, & \text{otherwise}.
    \end{cases}
\]
\end{definition}


\medskip

We now formalize the lifting operation that associates a binary circulant matrix to each polynomial in \(\mathbb{F}_2[D]/(D^L+1)\). This operation gives a finite circulant, or tail-biting, block representation of polynomial convolutional-code components.

\begin{definition}[Lifting operation for polynomials]\label{def:scalar-lift}
Fix $L\in\mathbb{N}$ and consider the ring $R_L = \mathbb{F}_2[D]/(D^L+1)$.  
For a polynomial $p(D)\in R_L$
the \emph{lifting operation} $\mathcal L_L$ maps $p(D)$ to an $L\times L$ binary circulant matrix
\[
  \mathcal L_L(p) \;\triangleq\; \sum_{i=0}^{L-1} a_i \, \mathrm{Circ}_L(i),
\]
where the sum is taken over $\mathbb{F}_2$ (entrywise modulo~2).
\end{definition}


Intuitively, the coefficient $a_i=1$ tells us to add the circulant permutation matrix corresponding to a shift by $i$ positions; if $a_i=0$, that shift is absent.  The first row of $\mathcal L_L(p)$ is exactly the coefficient vector $(a_0,a_1,\dots,a_{L-1})$, and all other rows are its cyclic shifts. Now we extend this lifting operation entrywise to polynomial matrices.

\begin{definition}[Polynomial–circulant correspondence]\label{def:poly-circ}
Let $P(D) = \bigl(p_{ij}(D)\bigr)\in R_L^{m\times n}$ be an $m\times n$ matrix with entries in $R_L$.  
The \emph{lifting} (or \emph{polynomial–circulant correspondence}) associates to $P(D)$ the binary matrix
\[
  \mathcal L_L\!\bigl(P(D)\bigr)
    \;\triangleq\;
    \bigl( \mathcal L_L(p_{ij}(D)) \bigr)_{i,j}
    \;\in\; \mathbb{F}_2^{\,mL \times nL},
\]
i.e., each polynomial entry $p_{ij}(D)$ is replaced by the $L\times L$ circulant matrix $\mathcal L_L\bigl(p_{ij}(D)\bigr)$ from Definition~\ref{def:scalar-lift}.  The resulting $mL\times nL$ matrix is block-circulant with $L\times L$ blocks.
\end{definition}



\subsection{Depolarizing Channel Model and BP-based Decoders}
\label{subsec:benchmark_decoders}

To evaluate the finite-length error-correction performance of our constructed codes, we will later use several standard BP-based decoders over the depolarizing channel. The quaternary belief propagation (QBP) decoder ~\cite{poulin2008iterative,kuo2020refined} performs message passing on the CSS Tanner graphs while keeping a local belief over the four Pauli symbols \(\{I,X,Y,Z\}\). Under the depolarizing channel, the error $q_i$ on the qubit \(i\) satisfies
\[
\Pr(q_i=I)=1-p,
\]
\[
\Pr(q_i=X)=\Pr(q_i=Y)=\Pr(q_i=Z)=\frac{p}{3}.
\]
Unlike independent \(X\)- and \(Z\)-component decoding, QBP accounts for the coupling introduced by the \(Y\) error. The decoder iteratively updates the messages and returns a hard Pauli estimate, which is then checked for syndrome
consistency and logical correctness.

In QBP with guided decimation (QBPGD) \cite{yao2024bpgd}, short runs ($T$ iterations) of the QBP decoder are interleaved with $G$ reliability-based decimation steps. At each decimation step, one highly reliable qubit is fixed, and QBP is then restarted on the reduced problem. This procedure helps break symmetric message passing failures caused by short cycles and degeneracy, but it can increase the decoding cost because several QBP runs may be required. We also evaluate our codes using a reinforcement-learning-based sequential decoder (RL-S) \cite{moradi2026rlbp}. It uses the same underlying BP message-update rules, but replaces the flooding schedule with a learned sequential variable-node update order. The learned policy is trained offline and then used during inference to select the next variable node based on a local syndrome-driven state. 

\section{The Conjugate Transpose Operation\\ in Quotient Rings}
We now define and establish basic algebraic properties of the conjugate transpose in $R_L=\mathbb{F}_2[D]/(D^L+1)$. This operation is crucial for our approach, since it underpins the orthogonality proofs and design of LP codes.
\subsection{Algebraic Properties}
By Definition \ref{def3}, polynomial conjugation is the map $p(D) \mapsto p^*(D) = p(D^{-1})$ interpreted in the quotient ring (with $D^L \equiv 1$). 


\subsection{Matrix Conjugate Operations}
\begin{lemma}[Basic properties of $^{*}$]
\label{lem:star-props}
The map $^{*}:R_L\to R_L$ is $\mathbb{F}_2$-linear, involutive, and multiplicative:
\[
(p+q)^{*}=p^{*}+q^{*},\qquad (p^{*})^{*}=p,\qquad (pq)^{*}=p^{*}q^{*}.
\]
\end{lemma}
\begin{proof}
Linearity and involution follow immediately from Definition \ref{def3}. For multiplicativity, note that
$^{*}$ acts as $D\mapsto D^{-1}$, and $(D^{-1})^k=D^{-k}$ in $R_L$.
\end{proof}

\begin{definition}[Matrix Conjugate Transpose]
For a polynomial matrix $B = [b_{ij}(D)]$ (with entries in $R_L$), the conjugate transpose $B^*$ is defined by
\[
\bigl(B^{*}\bigr)_{ij}(D) = \bigl(b_{ji}(D)\bigr)^{*}.
\]
That is, $ B^*$ is the transpose of $B$ with each entry replaced by its conjugate polynomial.
\end{definition}
We use the convention of the $^*$ operator notation for this operation, which is not to be confused with the Hermitian transpose. The conjugation is performed within the ring $R_L$ and with conjugation as defined in Def.~\ref{def3}.

\begin{proposition}[Conjugate Transpose Properties]
For any conformable polynomial matrices $A, B$:
\begin{itemize}
    \item $(A + B)^* = A^* + B^*$.
    \item $(A \cdot B)^* = B^* \cdot A^*$.
    \item $(A^*)^* = A$.
    \item $(A \otimes B)^* = A^* \otimes B^*$, where $\otimes$ denotes the Kronecker product.
\end{itemize}

\end{proposition}




\section{Lifted-Product Construction and Orthogonality}
\label{sec:lp}

\subsection{Lifting map and the conjugate-compatibility condition}
We use the standard circulant lifting map from $R_L$ to binary $L\times L$ circulant matrices following Definition \ref{def:scalar-lift}.

\begin{lemma}[Ring-homomorphism property]
\label{lem:ring-hom}
For all $p,q\in R_L$,
\[
 \mathcal L_L(p+q)= \mathcal L_L(p)+\mathcal L_L(q),\qquad \mathcal L_L(pq)=\mathcal L_L(p)\,\mathcal L_L(q),
\]
and the same holds entrywise for lifted matrices (block-circulant matrix multiplication).
\end{lemma}
\begin{proof}
This follows from a standard property of circulant polynomial representations: multiplication and addition of circulant matrices corresponds to multiplication and addition in $R_L$.
\end{proof}

\begin{definition}[Conjugate-compatible lifting (explicit condition)]
\label{def:conj-compatible}
We say that the lifting map $\mathcal L_L$ is \emph{conjugate-compatible} with $^{*}$ if, for every $p\in R_L$,
\[
\mathcal L_L(p^{*}) = \mathcal L_L(p)^{T}.
\]
\end{definition}

\begin{lemma}[Conjugate-compatibility holds for standard circulant lifting]
\label{lem:compat-holds}
The lifting map in Definition~\ref{def:scalar-lift} is conjugate-compatible: $\mathcal L_L(p^{*})=\mathcal L_L(p)^T$ for all $p\in R_L$.
Consequently, for any matrix $B$ over $R_L$,
\[
\mathcal L_L(B^{*}) = \mathcal L_L(B)^{T}.
\]
\end{lemma}
\begin{proof}
It suffices to verify the claim on monomials. For $p(D)=D^i$,
\[
p(D)^{*}=D^{-i}=D^{(L-i)\bmod L}.
\]
Also, $\mathrm{Circ}_L(i)^T=\mathrm{Circ}_L(-i)=\mathrm{Circ}_L((L-i)\bmod L)$.
Thus $\mathcal L_L((D^i)^{*}) = \mathcal L_L(D^i)^T$. This is then extended by $\mathbb{F}_2$-linearity to all $p\in R_L$ and then entrywise to matrices.
\end{proof}

\subsection{Lifted-product polynomial check matrices}
Let $B_1\in R_L^{m_1\times n_1}$ and $B_2\in R_L^{m_2\times n_2}$ be arbitrary polynomial parity-check base matrices.
Define the asymmetric LP polynomial check matrices
{\small
\begin{align}
B_X &\triangleq \Big[\, B_1\otimes I_{n_2}\ \Big|\ I_{m_1}\otimes B_2^{*}\,\Big]\in R_L^{\,m_1 n_2 \times \left(n_1 n_2 + m_1 m_2\right)}, \label{eq:BX-def}\\
B_Z &\triangleq \Big[\, I_{n_1}\otimes B_2\ \Big|\ B_1^{*}\otimes I_{m_2}\,\Big]\in R_L^{\,n_1 m_2 \times \left(n_1 n_2 + m_1 m_2\right)}, \label{eq:BZ-def}
\end{align}}
where $\otimes$ denotes the Kronecker product over the coefficient ring $R_L$ and $I_s$ is the $s\times s$ identity matrix. Then define the lifted binary parity-check matrices by
\[
\begin{aligned}
H_X &\triangleq \mathcal{L}_L(B_X)
\in \mathbb{F}_2^{\,m_1 n_2 L \times \left(n_1 n_2 + m_1 m_2\right)L}, \\
H_Z &\triangleq \mathcal{L}_L(B_Z)
\in \mathbb{F}_2^{\,n_1 m_2 L \times \left(n_1 n_2 + m_1 m_2\right)L}.
\end{aligned}
\]

\subsection{Conditional LP orthogonality theorem}
\begin{theorem}[Polynomial-domain LP orthogonality]
\label{thm:lp-orth}
Assume the lifting map $\mathcal L_L$ is conjugate-compatible (Definition~\ref{def:conj-compatible}).
Then for any base matrices $B_1,B_2$ over $R_L$, the lifted-product matrices satisfy
\[
H_X H_Z^{T} = 0 \quad \text{over } \mathbb{F}_2.
\]
\end{theorem}

\begin{proof}
We first prove the polynomial-domain identity
\begin{equation}
\label{eq:poly-orth}
B_X\,B_Z^{*} = 0
\quad \text{in } R_L.
\end{equation}
\noindent From \eqref{eq:BZ-def} we have
\[
B_Z^{*}
=
\begin{bmatrix}
(I_{n_1}\otimes B_2)^{*} \\
(B_1^{*}\otimes I_{m_2})^{*}
\end{bmatrix}
=
\begin{bmatrix}
I_{n_1}\otimes B_2^{*} \\
B_1\otimes I_{m_2}
\end{bmatrix},
\]
using $I_k^{*}=I_k$ and $(B_1^{*})^{*}=B_1$, and the standard identity $(A\otimes C)^{*}=A^{*}\otimes C^{*}$. Now multiply $B_X$ (a horizontal concatenation) by $B_Z^{*}$ (a vertical concatenation):
\[
B_XB_Z^{*}
=
(B_1\otimes I_{n_2})(I_{n_1}\otimes B_2^{*})
+
(I_{m_1}\otimes B_2^{*})(B_1\otimes I_{m_2}).
\]
Because the coefficient ring $R_L$ is commutative, the mixed-product rule for Kronecker products applies:
\[
(A\otimes C)(B\otimes D)=(AB)\otimes(CD),
\]
for all conformable matrices over $R_L$. Hence
\[
(B_1\otimes I_{n_2})(I_{n_1}\otimes B_2^{*})
=
(B_1I_{n_1})\otimes(I_{n_2}B_2^{*})
=
B_1\otimes B_2^{*},
\]
and similarly
\[
(I_{m_1}\otimes B_2^{*})(B_1\otimes I_{m_2})
=
(I_{m_1}B_1)\otimes(B_2^{*}I_{m_2})
=
B_1\otimes B_2^{*}.
\]
Therefore
\[
B_XB_Z^{*}=(B_1\otimes B_2^{*})+(B_1\otimes B_2^{*})=0
\]
because $\mathrm{char}(R_L)=2$. This proves \eqref{eq:poly-orth}.

Now we lift to binary matrices. By Lemma~\ref{lem:ring-hom},
\[
\mathcal L_L(B_XB_Z^{*}) = \mathcal L_L(B_X)\,\mathcal L_L(B_Z^{*}).
\]
By conjugate-compatibility (Lemma~\ref{lem:compat-holds}),\par
$\mathcal L_L(B_Z^{*})=\mathcal L_L(B_Z)^T=H_Z^T$.
Thus
\[
\begin{aligned}
H_XH_Z^{T}
&= \mathcal L_L(B_X)\,\mathcal L_L(B_Z)^{T}
 = \mathcal L_L(B_X)\,\mathcal L_L(B_Z^{*})\\
&= \mathcal L_L(B_XB_Z^{*}) = \mathcal L_L(0)
 = 0.
\end{aligned}
\]
\end{proof}\vspace{-0.5cm}

\begin{remark}
Theorem \ref{thm:lp-orth} shows that the CSS commutation condition can be certified before binary expansion. Instead of forming the full binary matrices $H_X$ and $H_Z$ and then checking $H_XH_Z^T=0$, one may verify the corresponding polynomial-domain identity over $R_L$. 
\end{remark}


\begin{example}[NASA constraint length 7 symmetric LP example]
Consider the NASA constraint length 7, rate $R=1/2$ convolutional code \cite{hamkins1999joint} with polynomial base matrix
\[
B(D)=\bigl[h_{11}(D)\;\;h_{12}(D)\bigr]
\in R_7^{1\times 2},
~~
R_7=\mathbb{F}_2[D]/(D^7+1),
\]
where
\[
h_{11}(D)=1+D+D^3+D^4+D^6,
\]
\[
h_{12}(D)=1+D^3+D^4+D^5+D^6.
\]
We take the symmetric lifted-product construction, so that
$B_1=B_2=B$, $m=1$, $n=2$.

By Definition~\ref{def3}, the conjugates are
\[
h_{11}^*(D)=h_{11}(D),
\quad
h_{12}^*(D)=1+D+D^2+D^3+D^4.
\]
Hence the polynomial LP check matrices are
\[
\begin{aligned}
B_X
&=
\bigl[\,
B\otimes I_2
\;\big|\;
I_1\otimes B^*
\,\bigr] \\
&=
\left[\begin{array}{cccc|c}
h_{11}(D) & 0 & h_{12}(D) & 0 & h_{11}^*(D)\\
0 & h_{11}(D) & 0 & h_{12}(D) & h_{12}^*(D)
\end{array}\right].
\end{aligned}
\]
and
\[
\begin{aligned}
B_Z
&=
\bigl[\,
I_2\otimes B
\;\big|\;
B^*\otimes I_1
\,\bigr] \\
&=
\left[\begin{array}{cccc|c}
h_{11}(D) & h_{12}(D) & 0 & 0 & h_{11}^*(D)\\
0 & 0 & h_{11}(D) & h_{12}(D) & h_{12}^*(D)
\end{array}\right].
\end{aligned}
\]

Now apply the standard circulant lifting map $\mathcal L_7(\cdot)$ entrywise to obtain the binary parity-check matrices
\[
H_X=\mathcal L_7(B_X),\qquad H_Z=\mathcal L_7(B_Z).
\]
Since $m=1$, $n=2$, and $L=7$, both lifted matrices have
$N=(n^2+m^2)L=(2^2+1^2)\cdot 7=35$
columns and
$mnL=1\cdot 2\cdot 7=14$ rows. Therefore, $H_X,H_Z\in \mathbb{F}_2^{14\times 35}$. Moreover, the standard circulant lift is conjugate compatible, i.e.,
\[
\mathcal L_7\!\left(p^*(D)\right)=\mathcal L_7\!\left(p(D)\right)^T, 
\quad \text{ for all }\>\> p(D)\in R_7.
\]
Therefore, by Theorem~\ref{thm:lp-orth}, $H_XH_Z^T=0$,
so the lifted pair $(H_X,H_Z)$ defines a valid CSS code.

\end{example}
This example is intended only to illustrate the polynomial-domain
construction in Theorem \ref{thm:lp-orth} and the conjugate-compatible lifting
argument that guarantees $H_XH_Z^T=0$.



\section{Numerical Results}
\label{sec:numerical}
In this section, we provide finite-length examples obtained from the polynomial domain LP construction of Theorem \ref{thm:lp-orth} and evaluate the resulting QLDPC codes using the benchmark decoders described in Section~\ref{subsec:benchmark_decoders}. For QBP and RL-S, we set a maximum number of iterations $T$ as noted in each figure. For QGBPD, we set $T$ =100 and $G=N$, for a total maximum of $GT$ iterations. We use
\[
B_1(D)=B_2(D)=H(D),
\]
where
\[
H(D)=
\begin{bmatrix}
D^{50} & D^{105} & D^{74} & D^4\\
D^8+D^{30} & D^7+D^{105} & D^{80} & D^{73}\\
D^{28} & D^{37} & D^{73} & D^{13}
\end{bmatrix}.
\]
The base matrix $H(D)$ was selected through a computational search over small $3\times 4$ polynomial matrices. We focused on full-support polynomial base matrices that remain sparse after circulant lifting while providing enough algebraic structure to produce nonzero-rate LP codes with balanced $X$- and $Z$-distance behavior. The selected base matrix gave the strongest finite-length results among the candidates considered. For each lifting size $L$, the entries of $H(D)$ are reduced modulo $D^L+1$, and the lifted-product CSS parity-check matrices are formed according to Theorem \ref{thm:lp-orth} with $B_1=B_2=H$. 

The parameters of the codes obtained from the available computations are summarized in Table~\ref{tab:lp_matrix1}. The table shows that changing the lifting size \(L\) produces a family of structured finite-length QLDPC codes with different block lengths, dimensions, and distance estimates. In the following simulations, we focus on selected members of this family and compare QBP, QBPGD, and RL-S over the depolarizing channel. For these decoders, $T$ is the maximum number of iterations allowed and in the case of QBPGD the total is multiplied by the rounds of decimation $G$.\footnote{In our experiments, we use a large $T$ for QBP and set $G=N$ and $T=100$ for QBPGD in order to ensure good performance at the cost of high complexity. This is used so that we can see the limiting behavior of these codes with the respective decoders. Many good strategies exist to lower that complexity, often maintaining similar performance.} The reported metric is the block error rate, {\it i.e.}, the probability that the decoder either fails to return a consistent syndrome or returns a recovery in the wrong logical coset.

%



\begin{table}[htbp]
\caption{Parameters of the lifted-product codes obtained from the \(3\times 4\) polynomial base matrix with \(B_1(D)=B_2(D)=H(D)\).}
\label{tab:lp_matrix1}
\centering
\begin{tabular}{c c c c c c}
\hline
\(L\) & \(N\) & \(K\) & \(d_X\) & \(d_Z\) & \(d\) \\
\hline
10 & 250  & 18 &  8 &  8 &  8 \\
11 & 275  & 15 & 12 & 12 & 12 \\
12 & 300  & 20 & 12 & 12 & 12 \\
13 & 325  & 17 & 14 & 14 & 14 \\
14 & 350  & 22 & 12 & 12 & 12 \\
15 & 375  & 19 & 16 & 16 & 16 \\
16 & 400  & 24 & 13 & 13 & 13 \\
17 & 425  & 21 & 18 & 18 & 18 \\
24 & 600  & 32 & \(\leq 20\) & \(\leq 20\) & \(\leq 20\) \\
41 & 1025 & 45 & \(\leq 26\) & \(\leq 26\) & \(\leq 26\) \\
83 & 2075 & 87 & \(\leq 34\) & \(\leq 34\) & \(\leq 34\) \\
\hline
\end{tabular}\vspace{0mm}
\end{table}

\begin{figure}[t]
  \centering
  \includegraphics[width=\linewidth]{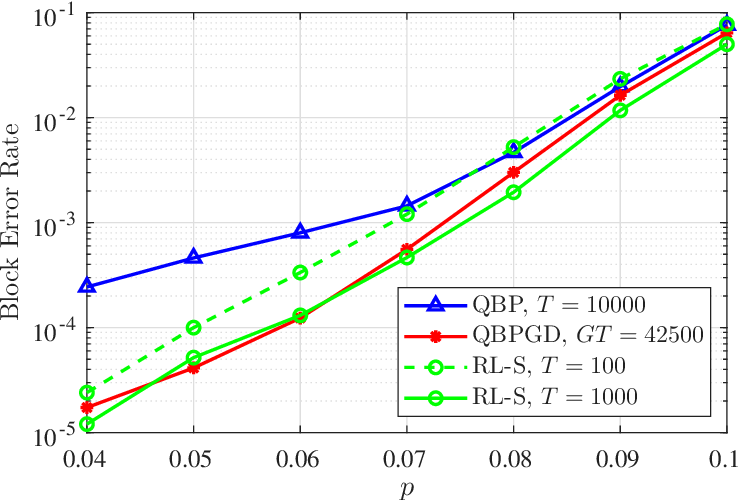}\vspace{-5mm}
  \caption{Block error rate comparison for the constructed
  \( [[425,21,18]] \) lifted-product QLDPC code obtained from the
  \(3\times 4\) polynomial base matrix with lifting size \(L=17\).}
  \label{fig:RL_LPCC_n425_FER}\vspace{1mm}
\end{figure}

Fig.~\ref{fig:RL_LPCC_n425_FER} shows the block error rate performance of our constructed \( [[425,21,18]] \) code corresponding to \(L=17\). We compare QBP, QBPGD, and RL-S with different iteration budgets. The results show that the constructed code can be decoded effectively by BP-based methods over the depolarizing channel. Fig.~\ref{fig:RL_LPCC_n600_FER} considers the constructed length \(N=600\) code obtained with lifting size \(L=24\). For this code, \(K=32\) and a minimum distance search implemented using the Gurobi Optimizer~\cite{gurobi} gives \(d_X,d_Z,d\leq 20\), as shown in Table~\ref{tab:lp_matrix1}. The same set of benchmark decoders is used over the depolarizing channel. The results show that the RL-S scheme improves the block error rate relative to QBP. In particular, RL-S with a larger iteration budget gives the best performance among the plotted decoders.

\begin{figure}[t]
  \centering
  \includegraphics[width=\linewidth]{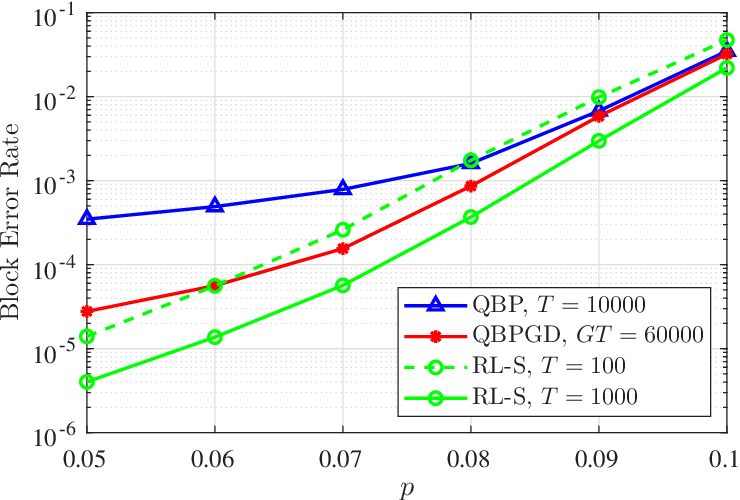}\vspace{-5mm}
  \caption{Block error rate comparison for the constructed
  \( [[600,32,d\leq 20]] \) lifted-product QLDPC code obtained from the
  \(3\times 4\) polynomial base matrix with lifting size \(L=24\).}
  \label{fig:RL_LPCC_n600_FER}\vspace{1mm}
\end{figure}

\begin{figure}[t]
  \centering
  \includegraphics[width=\linewidth]{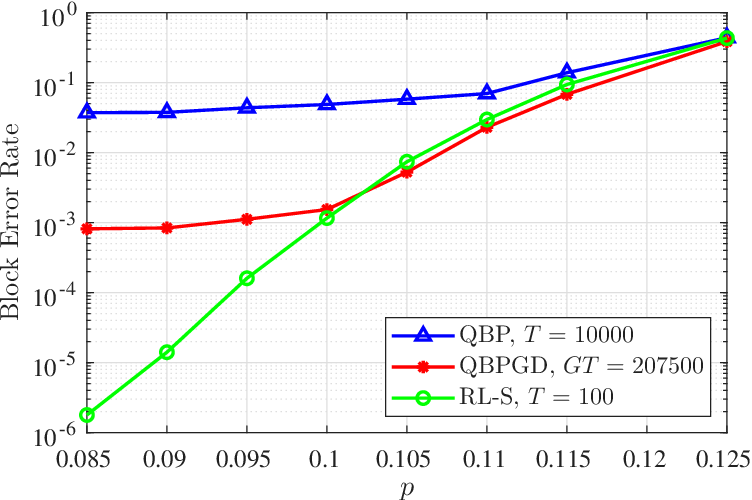}
 \vspace{-5mm}\caption{Block error rate comparison for the constructed
  \( [[2075,87,d\leq 34]] \) lifted-product QLDPC code obtained from the
  \(3\times 4\) polynomial base matrix with lifting size \(L=83\).}
  \label{fig:RLLPCC_n2075_FER}\vspace{1mm}
\end{figure}

Finally, Fig.~\ref{fig:RLLPCC_n2075_FER} evaluates a larger constructed code with block length \(N=2075\), corresponding to lifting size \(L=83\). For this code, \(K=87\) and a minimum distance search gives \(d_X,d_Z,d\leq 34\). For this construction, we observe that the QBP and QBPGD show a severe error floor, but that the RL-S decoder does not. This is due, in part, to the large number of 4-cycles in the graph ($30,378$ 4-cycles). The observed improvement is consistent with our previous results for RL decoding of polar codes \cite{moradi2025enhancing}, for which the graphs have large numbers of 4-cycles and the agent can effectively optimize the message passing schedule to mitigate that weakness. 


\section{Conclusions}

We presented a polynomial-domain formulation of finite-length lifted-product QLDPC codes over the quotient ring $R_L=\mathbb{F}_2[D]/(D^L+1).$ In this formulation, small polynomial base matrices are lifted entrywise to binary circulant blocks, giving a compact representation of the full binary parity-check matrices. The main algebraic contribution is the identification and use of conjugate-compatible circulant lifting. This property allows the CSS commutation condition \(H_XH_Z^T=0\) to be proved directly in the polynomial domain before binary expansion. This viewpoint connects classical convolutional code polynomial descriptions with finite-length LP-QLDPC code construction. We illustrated the construction using a constraint length 7 NASA convolutional code and a $3\times4$ polynomial base matrix, producing a family of structured finite-length QLDPC codes. Selected examples were evaluated over the depolarizing channel using BP-based decoders, showing that the constructed codes are compatible with practical iterative decoding methods.

This work should be viewed as a description and verification framework that can be used for new LP-QLDPC code design rather than a new asymptotic LP result. Future work includes combining this polynomial-domain formulation with systematic distance-search constraints, optimizing the choice of polynomial base matrices, and designing syndrome-extraction schedules adapted to the resulting structure.



\section*{Acknowledgement}
This material is based upon work supported by the National Science Foundation under Grant No. CCF-2145917. \pagebreak
\balance
\bibliographystyle{IEEEtran}
\bibliography{bibliography}

\appendices




\end{document}